\documentclass[11pt]{article}

\usepackage[T1]{fontenc}
\usepackage[sc]{mathpazo}
\usepackage{amsmath}
\usepackage{amssymb}
\usepackage{enumerate}
\usepackage{amsthm}
\usepackage{amsfonts,mathrsfs}
\usepackage[style]{fncychap}
\usepackage{graphicx} 
\usepackage{geometry} 
\usepackage{eepic}
\usepackage{ifthen}
\newboolean{ElectronicVersion}
\setboolean{ElectronicVersion}{true} 

\usepackage{hyperref}
\hypersetup{pdfpagemode=UseNone}

\newcommand{\tinyspace}{\mspace{1mu}}

\newcommand{\abs}[1]{\left\lvert\tinyspace #1 \tinyspace\right\rvert}

\newcommand{\setft}[1]{\mathrm{#1}}

\newcommand{\density}[1]{\setft{D}\left(#1\right)}
\newcommand{\unitary}[1]{\setft{U}\left(#1\right)}

\def\I{\mathbb{1}}

\newenvironment{mylist}[1]{\begin{list}{}{
    \setlength{\leftmargin}{#1}
    \setlength{\rightmargin}{0mm}
    \setlength{\labelsep}{2mm}
    \setlength{\labelwidth}{8mm}
    \setlength{\itemsep}{0mm}}}
    {\end{list}}

\def\ot{\otimes}

\newcommand{\out}[2]{| #1\rangle\langle #2 |}

\newcommand{\defeq}{\stackrel{\smash{\textnormal{\tiny def}}}{=}}

\newcommand{\pa}[1]{(#1)}
\newcommand{\Pa}[1]{\left(#1\right)}

\newcommand{\set}[1]{\{#1\}}
\newcommand{\Set}[1]{\left\{#1\right\}}

\newcommand{\ket}[1]{|#1\rangle}

\DeclareMathOperator{\trace}{Tr}
\newcommand{\ptr}[2]{\trace_{#1}\pa{#2}}
\newcommand{\Ptr}[2]{\trace_{#1}\Pa{#2}}

\newcommand{\Tr}[1]{\Ptr{}{#1}}

\def\cE{\mathcal{E}}
\def\cH{\mathcal{H}}

\def\rS{\mathrm{S}}

\newtheorem{thrm}{Theorem}[section]

\newtheorem{prop}[thrm]{Proposition}

\theoremstyle{definition}

\numberwithin{equation}{section}

\newcounter{questionnumber}

\begin{document}

\title{\Large The saturation of several universal inequalities in information-processing}

\author{Lin Zhang$^{1,}$\footnote{E-mail: godyalin@163.com;
linyz@zju.edu.cn}\ , Junde Wu$^{2,}$\footnote{Corresponding author.
E-mail:
wjd@zju.edu.cn}\ , Shao-Ming Fei$^3$\\
  {\small $^1$\it Institute of Mathematics, Hangzhou Dianzi University, Hangzhou 310018, PR~China}\\
  {\small $^2$\it Department of Mathematics, Zhejiang University, Hangzhou 310027, PR~China}\\
   {\small $^3$\it School of Mathematics of Sciences, Capital Normal University, Beijing 100048, PR China}}

\date{}
\maketitle \mbox{}\hrule\mbox\\
\begin{abstract}

In this paper, we characterize the saturation of four universal
inequalities in quantum information theory, including a variant
version of strong subadditivity inequality for von Neumann entropy,
the coherent information inequality, the Holevo quantity and average
entropy inequalities. These results shed new light on quantum
information inequalities.

\end{abstract}
\mbox{}\hrule\mbox\\

\section{Introduction and preliminaries}

Let $\cH$ be a finite dimensional complex Hilbert space. A
\emph{quantum state} $\rho$ on $\cH$ is a positive semi-definite
operator of trace one, in particular, for each unit vector
$\ket{\psi} \in \cH$, the operator $\rho = \out{\psi}{\psi}$ is said
to be a \emph{pure state}. The set of all quantum states on $\cH$ is
denoted by $\density{\cH}$. For each quantum state
$\rho\in\density{\cH}$, its von Neumann entropy is defined by
$\rS(\rho) = - \Tr{\rho\log_2\rho}$.

In quantum information theory, the strong subadditivity inequality
of von Neumann entropy, proved by Lie and Ruskai in \cite{Lieb}, stated
that
\begin{eqnarray*}\label{eq:SSA-1} \rS(\rho_{ABC}) +
\rS(\rho_B) \leqslant \rS(\rho_{AB}) + \rS(\rho_{BC}).
\end{eqnarray*}

The strong subadditivity inequality of von Neumann entropy is
ubiquitous, for example, as some direct consequences, the data
processing inequality, the well-known Holevo bound \cite{Roga}, in
particular, it connects with the monotonicity of relative entropy
under quantum channels \cite{Lindblad}.

If a reference system $D$ is introduced such that $\rho_{ABCD}$ is
the purification of $\rho_{ABC}$, that is, $\rho_{ABCD}$ is a pure
state with $\rho_{ABC} = \ptr{D}{\rho_{ABCD}}$, then an equivalent
version of strong subadditivity inequality of von Neumann entropy can be
described by \cite{Chuang}:
\begin{eqnarray}\label{eq:SSA-2}
\rS(\rho_D) +
\rS(\rho_B) \leqslant \rS(\rho_{AB}) + \rS(\rho_{AD}).
\end{eqnarray}

A \emph{quantum channel} $\Phi$ on $\cH$ is a trace-preserving
completely positive linear mapping defined on the set
$\density{\cH}$. It follows from (\cite[Prop.~5.2 and
Cor.~5.5]{Watrous}) that there exists linear operators
$\set{K_\mu}_\mu$ on $\cH$ such that $\sum_\mu K^\dagger_\mu K_\mu =
\I$ and for each quantum state $\rho$, we have the Kraus
representation
\begin{eqnarray*}
\Phi(\rho) = \sum_\mu K_\mu \rho K^\dagger_\mu.
\end{eqnarray*}

The \emph{complementary channel} $\widehat\Phi$ of $\Phi$ acts on
quantum state $\rho$ is defined by \cite{Roga}:
\begin{eqnarray*}
\widehat\Phi(\rho) = \sum_{\mu,\nu} \Tr{K_\mu \rho
K_\nu^\dagger}\out{\mu}{\nu},
\end{eqnarray*} the von Neumann entropy $\rS(\widehat\Phi(\rho))$
of $\widehat\Phi(\rho)$ is said to be the \emph{exchange entropy}.
For the topics related to complementary channel, the readers can
referred to \cite{Cubitt}.

The \emph{Coherent information $I_c(\rho,\Phi)$} can
be defined by the difference of von Neumann entropy
$\rS(\Phi(\rho))$ of output quantum state $\Phi(\rho)$ and the
exchange entropy \cite{Holevo}:
\begin{eqnarray*}
I_c(\rho,\Phi) = \rS(\Phi(\rho)) - \rS(\widehat\Phi(\rho)).
\end{eqnarray*}

It follows from the strong subadditivity
inequality of von Neumann entropy that the coherent information
$I_c(\rho,\Phi)$ can be bounded by the entropy $\rS(\rho)$ of the
initial state $\rS(\rho)$, that is
\begin{eqnarray}\label{eq:coh}
I_c(\rho,\Phi)\leqslant \rS(\rho).
\end{eqnarray}

Let $\cE = \set{(p_\mu,\rho_\mu)}$ be a quantum ensemble on $\cH$,
that is, each $\rho_\mu\in \density{\cH}$, $p_\mu > 0$, and
$\sum_{\mu} p_\mu=1$.  The \emph{Holevo quantity} of the quantum
ensemble $\Set{\Pa{p_\mu,\rho_\mu}}$ is defined by
\begin{eqnarray*}
\chi\Set{\Pa{p_\mu,\rho_\mu}} = \rS(\sum_\mu p_\mu \rho_\mu) -
\sum_\mu p_\mu \rS\Pa{\rho_\mu}.
\end{eqnarray*}

Let $\rho$ be a quantum state, and $\Phi (*)= \sum_\mu K_\mu(*)
K^\dagger_\mu$ be a quantum channel. If we denote
$$q_\mu = \Tr{K_\mu\rho K^\dagger_\mu}, \quad \quad \rho'_\mu =
q^{-1}_\mu K_\mu\rho K^\dagger_\mu
$$
and
$$
\rho' = \sum_\mu q_\mu \rho'_\mu,
$$ then $\Phi$ induced a quantum ensemble $\set{q_\mu,\rho'_\mu}$.

In \cite{Roga}, Roga proved that the Holevo quantity
$\chi(\set{q_\mu,\rho'_\mu})$ of quantum ensemble
$\set{q_\mu,\rho'_\mu}$ can be bounded by the exchange entropy
$\rS(\widehat\Phi(\rho))$, and the average entropy $\sum_\mu
q_\mu\rS(\rho'_\mu)$ of $\set{q_\mu,\rho'_\mu}$ can be bounded by
the entropy $\rS(\rho)$ of the initial state $\rho$, that is
\begin{eqnarray}\label{eq:exchangeentropy}
\chi(\set{q_\mu,\rho'_\mu}) \leqslant \rS\Pa{\widehat\Phi(\rho)},
\end{eqnarray}
and
\begin{eqnarray}\label{eq:initialentropy}
\sum_\mu q_\mu\rS(\rho'_\mu) \leqslant \rS(\rho).
\end{eqnarray}

In \cite{Hayden}, the saturation of strong subadditivity inequality $\rS(\rho_{ABC}) +
\rS(\rho_B) \leqslant \rS(\rho_{AB}) + \rS(\rho_{BC})$ of von Neumann entropy is presented:
\begin{prop}[\cite{Hayden}]\label{prop:SSAwithequality}
A state $\rho_{ABC}\in\density{\cH_A\ot\cH_B\ot\cH_C}$ saturates the strong subadditivity inequality of von Neumann entropy, that is, \begin{eqnarray*}
\rS(\rho_{AB}) + \rS(\rho_{BC}) = \rS(\rho_{ABC}) + \rS(\rho_B)
\end{eqnarray*}
if and only if there is a decomposition of system $B$ as
\begin{eqnarray*}
\cH_B = \bigoplus_j \cH_{b^L_j}\ot\cH_{b^R_j},
\end{eqnarray*} such that
\begin{eqnarray*}
\rho_{ABC} = \bigoplus_j \lambda_j\rho_{Ab^L_j}\ot \rho_{b^R_jC},
\end{eqnarray*}
where $\rho_{Ab^L_j}\in\density{\cH_A\ot\cH_{b^L_j}}$,
$\rho_{b^R_jC}\in\density{\cH_{b^R_j}\ot\cH_C}$,
$\set{\lambda_j}$ is a probability distribution.
\end{prop}

Let $\rho_{BC}\in\density{\cH_B\ot\cH_C}$,
$\rho_B =\ptr{C}{\rho_{BC}}, \rho_C =\ptr{B}{\rho_{BC}}$. The famous Araki-Lieb inequality showed that
\begin{eqnarray*}
\abs{\rS(\rho_B)-\rS(\rho_C)}\leqslant \rS(\rho_{BC}).
\end{eqnarray*}

In \cite{Zhang}, the saturation of Araki-Lieb inequality is presented:

\begin{prop}[\cite{Zhang}]\label{prop:Araki-Lieb}
$\rS(\rho_{BC}) = \rS(\rho_B) - \rS(\rho_C)$ if
and only if
\begin{enumerate}[(i)]
\item\label{1} $\cH_B$ can be decomposed into $\cH_B = \cH_L \ot \cH_R$,
\item\label{2} $\rho_{BC} = \rho_L \ot \out{\psi}{\psi}_{RC}$ for $\ket{\psi}_{RC} \in \cH_R \ot \cH_C$.
\end{enumerate}
\end{prop}

In \cite{Xi}, the authors applied Proposition 1.2 to study the
saturation of the upper bound of quantum discord. In \cite{Carlen},
the authors gave  Proposition 1.2 an elementary proof.

In this paper, we study the saturation of the four universal
inequalities \eqref{eq:SSA-2} through \eqref{eq:initialentropy}.

\section{The saturation of strong subadditivity inequality}

In this section, we give a characterization to the structure of
states which saturate the strong subadditivity inequality \eqref{eq:SSA-2} of von Neumann entropy. That is

\begin{thrm}\label{th:mainresult}
Let $\sigma_{ABC}\in\density{\cH_A\ot\cH_B\ot\cH_C}$. Then
\begin{eqnarray}\label{eq:anotherversion-SSA-1}
\rS(\sigma_A) + \rS(\sigma_C) = \rS(\sigma_{AB}) + \rS(\sigma_{CB})
\end{eqnarray}
if and only if there are two decompositions of system $A$ and $C$, respectively, as
\begin{eqnarray}
\cH_A =
\bigoplus_{i=1}^{K_A}\cH_{a^L_i}\ot\cH_{a^R_i}\quad\text{and}\quad
\cH_C = \bigoplus_{j=1}^{K_C}\cH_{c^L_j}\ot\cH_{c^R_j}
\end{eqnarray}
such that
\begin{eqnarray}\label{eq:abc}
\sigma_{ABC} = \bigoplus_{i,j}\mu_{ij}\sigma_{a^L_iBc^L_j} \ot
\sigma_{a^R_ic^R_j},
\end{eqnarray}
where
$\sigma_{a^L_iBc^L_j}\equiv\out{\psi}{\psi}_{a^L_iBc^L_j}\in\density{\cH_{a^L_i}\ot\cH_B\ot\cH_{c^L_j}}$
,$\sigma_{a^R_ic^R_j}\in\density{\cH_{a^R_i}\ot\cH_{c^R_j}}$ and
$\set{\mu_{ij}}$ is a joint probability distribution.
\end{thrm}

\begin{proof}
We introduce a reference system $D$ such that $\sigma_{ABCD}$ is a
purification of $\sigma_{ABC}$. Thus
Equation ~\eqref{eq:anotherversion-SSA-1} can be rewritten into
\begin{eqnarray}\label{eq:anotherversion-SSA-2}
\rS(\sigma_A) + \rS(\sigma_C) = \rS(\sigma_{CD}) + \rS(\sigma_{AD}).
\end{eqnarray}
It can be seen that, when the systems $A$ and $C$ are fixed, the
systems $B$ and $D$ play a symmetric role in
Equation ~\eqref{eq:anotherversion-SSA-1} and
Equation ~\eqref{eq:anotherversion-SSA-2}. Analogously, we have
\begin{eqnarray}
\rS(\sigma_A) + \rS(\sigma_{ABD}) &=& \rS(\sigma_{AB}) + \rS(\sigma_{AD}),\label{eq:anotherversion-SSA-3}\\
\rS(\sigma_{CBD}) + \rS(\sigma_C) &=& \rS(\sigma_{CD}) +
\rS(\sigma_{CB}).\label{eq:anotherversion-SSA-4}
\end{eqnarray}
Again, when the systems $B$ and $D$ are fixed, the systems $A$ and
$C$ play a symmetric role in Equation ~\eqref{eq:anotherversion-SSA-3} and
Equation ~\eqref{eq:anotherversion-SSA-4}.

Now it follows from Proposition~\ref{prop:SSAwithequality}~that there are
two decompositions of $A$ and $C$, respectively,
\begin{eqnarray}
\cH_A = \bigoplus_{i=1}^{K_A}
\cH_{a^L_i}\ot\cH_{a^R_i}\quad\text{and}\quad \cH_C =
\bigoplus_{j=1}^{K_C} \cH_{c^L_j}\ot\cH_{c^R_j}
\end{eqnarray}
such that
\begin{eqnarray}
\sigma_{ABD} = \bigoplus_i p_i\sigma_{a^L_iB}\ot
\sigma_{a^R_iD}\quad\text{and}\quad \sigma_{BCD} = \bigoplus_j
q_j\sigma_{Bc^L_j}\ot \sigma_{c^R_jD}.
\end{eqnarray}
Thus $\sigma_{ABC}$ must be of the form:
$$
\sigma_{ABC} = \bigoplus_{i,j} \mu_{ij}\sigma_{a^L_iBc^L_j}^{(ij)}
\ot \sigma_{a^R_ic^R_j}^{(ij)},
$$
where
$$
\rS\Pa{\sigma_{a^L_iB}^{(ij)}} + \rS\Pa{\sigma_{Bc^L_j}^{(ij)}} =
\rS\Pa{\sigma_{a^L_i}^{(ij)}} +
\rS\Pa{\sigma_{c^L_j}^{(ij)}}\quad(\forall i,j).
$$
Without loss of generality, we assume that the system $a^L_i$ and
$c^L_j$ can not decomposed like the $\cH_A$ and $\cH_C$,
respectively. Therefore $\sigma_{a^L_iBc^L_j}$ must be a pure state,
which implies that
$$
\sigma_{a^L_iBc^L_j} \equiv \out{\psi}{\psi}_{a^L_iBc^L_j}.
$$
Conversely, if the state $\sigma_{ABC}$ has the form of
Equation ~\eqref{eq:abc}, then it is easy to check that
Equation ~\eqref{eq:anotherversion-SSA-1} holds.
\end{proof}

We would like to point out that if $\sigma_{ABC}$ is a pure state,
by Proposition~\ref{prop:SSAwithequality}, then there must exist a
decomposition of $\sigma_{ABC}$ such that its substates are locally
pure states.

\section{The saturation of coherent information inequality}

In this section, we make an attempt towards the saturation of
coherent information inequality \eqref{eq:coh}.

\begin{thrm}\label{prop:coh-max}
Let $\rho\in\density{\cH}$ and $\Phi$ be a quantum channel defined
over $\cH$. Then $I_c(\rho,\Phi) = \rS(\rho)$ if and only if the
following statements hold:
\begin{enumerate}[(i)]
\item The Hilbert space $\cH$ can be decomposed into $\cH = \cH_L\ot\cH_R$;
\item The output state $\Phi(\rho)$ of the quantum channel $\Phi$ is a product state: $\Phi(\rho) = \rho_L\ot\rho_R$, where
$\rho_L\in\density{\cH_L},\rho_R\in\density{\cH_R}$.
\end{enumerate}
\end{thrm}

\begin{proof}
Note that
\begin{eqnarray}
\rS(\widehat\Phi(\rho)) = \rS((\I_A\ot\Phi)(\out{\mathbf{u}_\rho}{\mathbf{u}_\rho})),
\end{eqnarray}
where $\ket{\mathbf{u}_\rho}$ is a purification of $\rho$ in a larger Hilbert space $\cH_A\ot\cH_B$, where $\cH_B\equiv\cH$. It was shown that there exists a quantum channel $\Psi$ (see
\cite{Hayden}) such that
\begin{eqnarray}
I_c(\rho,\Phi) = \rS(\rho) \Longleftrightarrow
(\I_A\ot\Psi\circ\Phi)(\out{\mathbf{u}_\rho}{\mathbf{u}_\rho}) =
\out{\mathbf{u}_\rho}{\mathbf{u}_\rho}.
\end{eqnarray}
By the \emph{Stinespring dilation theorem} (see \cite{Watrous}), we may assume that
$$
\Phi(\rho) =
\Ptr{C}{U(\rho\ot\out{0}{0})U^\dagger},\quad
U\in\unitary{\cH_B\ot\cH_C}, \ket{0}\in\cH_C,
$$
which indicates that
\begin{eqnarray}
\I_A\ot\Phi(\out{\mathbf{u}_\rho}{\mathbf{u}_\rho}) &=&
\ptr{C}{(\I_A\ot
U)(\out{\mathbf{u}_\rho}{\mathbf{u}_\rho}\ot\out{0}{0})(\I_A\ot
U)^\dagger} \nonumber\\
&= &\Ptr{C}{\out{\Omega}{\Omega}},
\end{eqnarray}
where
$\ket{\Omega} \defeq (\I_A\ot U)(\ket{\mathbf{u}_\rho} \ot
\ket{0})$. Now
$$
\out{\Omega}{\Omega} =(\I_A\ot
U)(\out{\mathbf{u}_\rho}{\mathbf{u}_\rho} \ot
\out{0}{0})(\I_A\ot U)^\dagger
$$
is a tripartite state
on $\cH_A \ot \cH_B \ot \cH_C$, it follows that
\begin{eqnarray*}
\ptr{C}{\out{\Omega}{\Omega}} & = & \I_A\ot\Phi(\out{\mathbf{u}_\rho}{\mathbf{u}_\rho}) \equiv \Omega_{AB},\\
\ptr{A}{\out{\Omega}{\Omega}} & = & U(\rho\ot\out{0}{0})U^\dagger \equiv \Omega_{BC},\\
\ptr{AC}{\out{\Omega}{\Omega}} & = & \Phi(\rho) \equiv \Omega_B,
\end{eqnarray*}
where $\Omega_{ABC} \equiv \out{\Omega}{\Omega}$. From the above
expressions, it is obtained that
\begin{eqnarray*}
\rS(\Omega_{ABC}) & = & 0, \\
\rS(\Omega_B) & = & \rS(\Phi(\rho))\\
\rS(\Omega_{BC}) & = & \rS(\rho),\\
\rS(\Omega_{AB}) & = &
\rS((\I_A\ot\Phi)(\out{\mathbf{u}_\rho}{\mathbf{u}_\rho}))
\end{eqnarray*}
Apparently, $I_c(\rho,\Phi) = \rS(\rho) \Longleftrightarrow
\rS(\Phi(\rho)) =
\rS((\I_A\ot\Phi)(\out{\mathbf{u}_\rho}{\mathbf{u}_\rho})) +
\rS(\rho)$, that is,
\begin{eqnarray*}
I_c(\rho,\Phi) = \rS(\rho) &\Longleftrightarrow&
\rS(\Omega_B) = \rS(\Omega_{AB}) + \rS(\Omega_{BC})\\
&\Longleftrightarrow& \rS(\Omega_B) - \rS(\Omega_C)
=\rS(\Omega_{BC}).
\end{eqnarray*}
It follows from Proposition~\ref{prop:Araki-Lieb}~that this equality holds
if and only if
\begin{enumerate}[(i)]
\item $\cH_B$ can be factorized into the form $\cH_B = \cH_L \ot \cH_R$,
\item $\Omega_{BC} = \rho_L \ot \out{\psi}{\psi}_{RC}$ for $\ket{\psi}_{RC} \in \cH_R \ot \cH_C$.
\end{enumerate}
That is,
$$
\Phi(\rho) = \Ptr{C}{\rho_L \ot \out{\psi}{\psi}_{RC}}= \rho_L\ot\rho_R.
$$
This indicates that if the coherent information arrives at
$\rS(\rho)$, then the output state $\Phi(\rho)$ of the quantum
channel $\Phi$ is a product state.
\end{proof}

\section{The saturation of Holevo quantity and average entropy inequalities}

In this section, we study the saturation of Holevo quantity
inequality \eqref{eq:exchangeentropy} and average entropy inequality
\eqref{eq:initialentropy} which are induced by a quantum channel.

\begin{thrm}
With the above notation, we have the following result:
\begin{eqnarray}\label{eq:exchange-entropy}
\chi(\set{q_\mu,\rho'_\mu}) = \rS\Pa{\widehat\Phi(\rho)}~~\text{if
and only if}~~\Phi(\rho) = \bigoplus_{i,j} p_{ij}
\omega_{a^L_i}^{(ij)} \ot \omega_{a^R_i}^{(ij)}.
\end{eqnarray}
\end{thrm}

\begin{proof}
In order to prove the conclusion, we need to go back to the original
proof of \cite{Roga}. Where the authors introduced a tripartite state
\begin{eqnarray}
\omega_{ABC}
\defeq \sum_{\mu,\nu} \Pa{K_\mu\rho K^\dagger_\nu}_A \ot \out{\mu}{\nu}_B\ot
\out{\mu}{\nu}_C.
\end{eqnarray}
From the above expression, we see that $\omega_{ABC}$ is a symmetric
state for $BC$ relative to $A$ and $\cH_B=\cH_C$. Denote $q_\mu =
\Tr{K_\mu \rho K^\dagger_\mu}$ and $\rho'_\mu = q_\mu^{-1}K_\mu \rho
K^\dagger_\mu$. Since
$$
\rS(\omega_{BC}) =
\rS(\widehat\Phi(\rho)), \quad\rS(\omega_A) = \rS(\sum_\mu q_\mu
\rho'_\mu), \quad\sum_\mu q_\mu\rS(\rho'_\mu) = \rS(\omega_{AC}) -
\rS(\omega_B).
$$
Therefore $\chi(\set{q_\mu,\rho'_\mu}) = \rS\Pa{\widehat\Phi(\rho)}$
if and only if $\rS(\omega_A) + \rS(\omega_C) = \rS(\omega_{AB}) +
\rS(\omega_{BC})$. This amounts to say, by
Theorem~\ref{th:mainresult}, that
\begin{eqnarray}
\omega_{ABC} = \sum_{\mu,\nu} \Pa{K_\mu\rho K^\dagger_\nu}_A \ot \out{\mu}{\nu}_B\ot \out{\mu}{\nu}_C
= \bigoplus_{i,j} p_{ij} \omega_{a^L_iBc^L_j}^{(ij)}
\ot \omega_{a^R_ic^R_j}^{(ij)},
\end{eqnarray}
where each $\omega_{a^L_iBc^L_j}^{(ij)}$ is a pure state. Since both
$B$ and $C$ are identical, it follows that
$$
\sum_{\mu} K_\mu\rho K^\dagger_\mu\ot \out{\mu}{\mu} =
\bigoplus_{i,j} p_{ij} \omega_{a^L_ic^L_j}^{(ij)} \ot
\omega_{a^R_ic^R_j}^{(ij)},
$$
which implies that
$$
\Phi(\rho) = \omega_A = \bigoplus_{i,j} p_{ij} \omega_{a^L_i}^{(ij)}
\ot \omega_{a^R_i}^{(ij)}.
$$
From this expression, it is seen that the output state of a quantum
channel $\Phi$ with input state $\rho$ is a weighted state of block
diagonal form in some basis.
\end{proof}

\begin{thrm}
The average entropy attains the entropy of the initial state $\rho$,
that is, $\sum_\mu q_\mu\rS(\rho'_\mu) = \rS(\rho)$ if and only if
$\Phi(*)=U(*)U^\dagger$ for some unitary operator $U$.
\end{thrm}
\begin{proof}
Note that $\sum_\mu
q_\mu\rS(\rho'_\mu) = \rS(\rho)$ if and only if $\rS(\omega_{AC}) -
\rS(\omega_B) = \rS(\omega_{ABC}) = \rS(\omega_{AB}) -
\rS(\omega_C)$. This amounts to say, by
Proposition~\ref{prop:Araki-Lieb}, that
$$
\omega_{ABC} = \omega_L\ot\out{\psi}{\psi}_{RC},\quad \omega_{ACB} =
\omega_{\hat L}\ot\out{\psi}{\psi}_{\hat RB}.
$$
Since $\omega_{ABC}=\omega_{ACB}$, in fact,  the composite system of
$B$ and $C$ stays in a symmetric state when we ignore system $A$,
that is, swapping the role of $B$ and $C$ leaves $\omega_{ABC}$
invariant, it follows that
$$
\omega_A = \omega_L = \omega_{\hat L},\quad \omega_R = \omega_B,\quad \omega_C = \omega_{\hat R}.
$$
This fact indicates that system $A$ has no correlation with the
composite system of $B$ and $C$, which is in a pure symmetric state
$\out{\psi}{\psi}_{BC}$ on $\cH_B\ot\cH_C$ with $\cH_B=\cH_C$. Thus,
we have $\omega_{ABC} = \omega_A\ot\omega_{BC}$ with
$\omega_{BC}=\out{\psi}{\psi}_{BC}$, that is,
$$
\out{\psi}{\psi}_{BC} = \sum_{\mu,\nu} \Tr{K_\mu\rho K^\dagger_\nu} \out{\mu}{\nu}\ot\out{\mu}{\nu}.
$$
Now $W\ket{\mu}=\ket{\mu\mu}$ for all $\mu$ defines an isometry $W$,
and
$$
\out{\psi}{\psi}_{BC} = W\Pa{\sum_{\mu,\nu} \Tr{K_\mu\rho
K^\dagger_\nu} \out{\mu}{\nu}}W^\dagger.
$$
Therefore, the von Neumann entropy of $\sum_{\mu,\nu} \Tr{K_\mu\rho
K^\dagger_\nu} \out{\mu}{\nu}$ is vanished since isometric
transformation leaves the von Neumann entropy of state invariant.
This indicates that $$\sum_{\mu,\nu} \Tr{K_\mu\rho K^\dagger_\nu}
\out{\mu}{\nu} \equiv \widehat\Phi(\rho)$$ is still a pure state.
Define an isometry $V$ as follows:
$$
V\ket{\phi} = K_\mu\ket{\phi}\ot\ket{\mu}, ~~\forall \mu.
$$
It follows that
$$
V\rho V^\dagger = \sum_{\mu,\nu} K_\mu\rho K^\dagger_\nu\ot
\out{\mu}{\nu},
$$
implying that $\Phi(\rho) = \ptr{2}{V\rho V^\dagger}$ and
$\widehat\Phi(\rho)=\ptr{1}{V\rho V^\dagger}$. Note that
$\widehat\Phi(\rho)$ is a pure state, the bipartite state
$\sum_{\mu,\nu} K_\mu\rho K^\dagger_\nu\ot \out{\mu}{\nu}$ is a
product state: $V\rho V^\dagger = \ptr{2}{V\rho
V^\dagger}\ot\ptr{1}{V\rho V^\dagger}$. Therefore
\begin{eqnarray}\label{eq:pure-Envir}
\sum_{\mu,\nu} K_\mu\rho K^\dagger_\nu\ot \out{\mu}{\nu} =
\Phi(\rho)\ot\widehat\Phi(\rho).
\end{eqnarray}
Again since $\widehat\Phi(\rho)$ is a pure state, there must exist
complex numbers $\lambda_\mu$ such that $\Tr{K_\mu\rho
K^\dagger_\nu}=\lambda_\mu\bar\lambda_\nu$ for complex numbers
$\lambda_\mu$. Clearly $\sum_\mu\abs{\lambda_\mu}^2=1$. Now we can
infer from Equation ~\eqref{eq:pure-Envir} that
$$
\Phi(\rho) = \Pa{\lambda^{-1}_\mu K_\mu}\rho\Pa{\lambda^{-1}_\nu K_\nu}^\dagger= \Pa{\lambda^{-1}_\nu K_\nu}\rho\Pa{\lambda^{-1}_\mu K_\mu}^\dagger,\quad\forall \mu,\nu
$$
or
$$
K_\mu\rho K^\dagger_\nu = \lambda_\mu\bar\lambda_\nu\Phi(\rho),
$$
which implies that
\begin{eqnarray}
\rho &=& (\sum_\mu K_\mu^\dagger K_\mu)\rho(\sum_\nu K_\nu^\dagger K_\nu) = \sum_{\mu,\nu} K_\mu^\dagger (K_\mu\rho K_\nu^\dagger) K_\nu\\
&=& (\sum_\mu \lambda_\mu K^\dagger_\mu)\Phi(\rho) (\sum_\nu \lambda_\nu K^\dagger_\nu)^\dagger \equiv M\Phi(\rho)M^\dagger,
\end{eqnarray}
where $M\defeq \sum_\mu \lambda_\mu K^\dagger_\mu$. From Equation ~\eqref{eq:pure-Envir}, we can see that
$$
\rS(\rho) = \rS(V\rho V^\dagger)= \rS(\Phi(\rho)) +
\rS(\widehat\Phi(\rho)) = \rS(\Phi(\rho))
$$
since $\widehat\Phi(\rho)$ is a pure state. Moreover, we can have $I_c(\rho,\Phi) = \rS(\rho)$. This showed that
$$
\Phi(\rho) = \rho_L\ot\rho_R.
$$

In the above process, the output state of the complementary channel is
pure state. Without loss of generality, we assume that the
environment starts in a pure state, this implies that the
complementary channel is an unitary channel. From the basic
properties of quantum channel, we obtain that $\sum_\mu q_\mu\rS(\rho'_\mu) = \rS(\rho)$ if and only if the
quantum channel $\Phi$ is the unitary channel.
\end{proof}

\section{Concluding remarks}

In this paper we give characterizations of several famous quantum
information inequalities when becoming equalities. Specifically, we
characterize the saturation of four universal inequalities in
quantum information theory, including a variant version of strong
subadditivity inequality for von Neumann entropy, the coherent
information inequality, the Holevo quantity and average entropy
inequalities.  The proofs are based on the works of Hayden etc.
\cite{Hayden}, and that of Zhang and Wu \cite{Zhang}. These results
shed new light on quantum information inequalities.

In the future research, we will investigate the approximate version
of the above-discussed information inequalities since the
approximate information inequalities are more or less related to
entanglement theory.

\noindent\textbf{Acknowledgement.} We want to express our heartfelt
thanks to Wojciech Roga for useful comments. This work is supported
by National Natural Science Foundation of China (11301124, 11171301)
and the Doctoral Programs Foundation of Ministry of Education of
China (J20130061).



\end{document}